\documentclass[12pt,twoside]{article}
\usepackage{amsmath,amsfonts,amssymb,amsthm}
\usepackage{bbm}
\usepackage{graphicx}
\usepackage{hyperref}
\usepackage{color}
\newcommand{\grad}{\nabla}

\newcommand{\laplace}{\Delta}

\renewcommand{\div}{\grad\cdot}

\newcommand{\R}{\mathbbm{R}}

\newcommand{\Ha}{\ensuremath{\mathcal{H}}}

\newcommand{\eps}{\varepsilon}

\newtheorem{theorem}{Theorem}
\newtheorem{lemma}{Lemma}
\newtheorem{remark}{Remark}
\newtheorem{definition}{Definition}

\newcommand{\tacka}{\, \cdot\,}

\begin{document}

\title{Maximal mixing by incompressible fluid flows}

\author{Christian Seis\footnote{Department of Mathematics, University of Toronto, 40 St.\ George Street, M5S 2E4, Toronto, Ontario, Canada}}

\maketitle

\begin{abstract}We consider a model for mixing binary viscous fluids under an incompressible flow. We proof the impossibility of perfect mixing in finite time for flows with finite viscous dissipation. As measures of mixedness we consider a Monge--Kantorovich--Rubinstein transportation distance and, more classically, the $H^{-1}$ norm. We derive rigorous a priori lower bounds on these mixing norms which show that mixing cannot proceed faster than exponentially in time. The rate of the exponential decay is uniform in the initial data.
\end{abstract}

\section{Introduction}

The present manuscript is concerned with optimal stirring strategies for binary mixtures of incompressible viscous fluids. More precisely, we study decay rates of certain mixing norms with respect to a constrained velocity field. We focus on passive scalar mixing, which means that the feedback of the transported quantity on the flow field is negligible. To model the binary mixture, we consider an indicator function $\rho=\rho(t,x)$ which takes the values $+1$ and $-1$ only, so that the sets $\{\rho=1\}$ and $\{\rho=-1\}$ represent the regions in which the fluid consists of component `A' and component `B', respectively. As usual, $t$ and $x$ are the time and space variable, respectively. The stirring velocity field will be denoted by $u=u(t,x)$, and we assume this vector field to be smooth. The transport of the passive scalar by the incompressible flow is then described by the system
\begin{eqnarray}
\partial_t \rho +u\cdot\grad\rho &=&0,\label{8}\\
\div u&=&0,\label{9}
\end{eqnarray}
and we impose the initial condition $\rho(0,x)=\rho_0(x)\in\{\pm1\}$. For mathematical convenience, we finally assume that all quantities are periodic in the spatial variables with period cell $[0,1)^d$. Observe that the total mass of each species is preserved under the flow. In the case of a critical mixture, that is the case where both components have the same volume fraction and which we assume in this paper, this means
\[
\int \rho(t,\tacka)\, dx\;=\;\int\rho_0\, dx\;=\;0\quad\mbox{for all }t>0.
\]

\medskip

A natural question, first formulated in \cite{LinThiffeaultDoering11}, addresses the optimality of stirring strategies: How efficient can a binary fluid be mixed by an incompressible flow? Classical measures of the mathematical quantification of the degree of mixedness are negative Sobolev norms, with the homogeneous $H^{-1}$ norm leading the way --- not only because of its analytical accessibility. While any $\dot H^{-s}$ norm with $s>0$ can be considered as a mixing norm, the $\dot H^{-1}$ is distinguished as it scales like a length, and thus, the actual value of this norm relates to the typical size of one-component (i.e., unmixed) regions in the mixture. Negative Sobolev norms measure oscillations of the indicator function: The larger the typical length scales, the larger the negative Sobolev norm. The actual value of $s$ corresponds to a weight in the measurement of the length scales. On the upper end, the $H^0$ (i.e., $L^2$) norm is the variance, which is conserved under the flow \eqref{8}\&\eqref{9}, and reveals thus no information on the domain morphology. In the context of mixing, negative Sobolev norms were used as large-scale mixing measures in \cite{MathewMezicPetzold05,DoeringThiffeault06,ShawThiffeaultDoering07,LinThiffeaultDoering11}. Alternative measures of the degree of mixedness can be borrowed from the theory of optimal transportation: Monge--Kantorovich--Rubinstein (MKR) distances can be used to measure the ``distance'' between two configurations, for instance, the current state and the perfectly mixed state. These MKR distances have been successfully used as measures of mixedness in related studies on demixing problems \cite{KohnOtto02,OttoRumpSlepcev06,Slepcev08,BOS,OSS}. We will review the definition of the MKR mixing measure in Subsection \ref{S2} below. For more notions of mixing measures, we refer to \cite{LinThiffeaultDoering11,Lunasin12} and references therein.

\medskip

Any quantitative statement about the efficiency of stirring strategies necessitates to specify certain constraints on the flow, for instance, by fixing the budget for the kinetic energy $\|u\|_{L^2}$, the viscous dissipation $\|\grad u\|_{L^2}$, or the palenstrophy $\|\laplace u\|_{L^2}$. Under such constraints, several optimal control problems were proposed to study effective mixing, e.g., \cite{MathewMezicGrivopoulos07,Liu08,BubanovCortelezzi10}. In practice, however, solving optimal control problems may be difficult, so different approaches that are suitable also for numerical studies are desirable. The authors in \cite{LinThiffeaultDoering11,Lunasin12} seek to find optimal stirring strategies by choosing the velocity field $u$ in each time step in direction of the steepest descent of the mixing norms. Their strategy also includes identifying absolute limits on how fast perfect mixing could ever be achieved subject to given constraints on the stirring flow. Here, the choice of the constraint is crucial: While optimal stirring strategies with a fixed energy budget $\|u\|_{L^2}$ yield to perfect mixing in finite time \cite[Sec.\ 2]{Lunasin12}, the $\dot H^{-1}$ mixing norm cannot decay faster than exponentially when stirring under a palenstrophy constraint $\|\laplace u\|_{L^2}$, cf.\ \cite[Sec.\ 3]{Lunasin12}. Such statements on the limit of stirring strategies result from rigorous {\em a priori} lower bounds on the mixing norm. In the present work, we address mixing with a prescribed viscous dissipation $\|\grad u\|_{L^2}$ (or even more general $\|\grad u\|_{L^p}$ with $p\in(1,\infty]$). This case is most natural for engineering applications where the focus is on overcoming viscous dissipation to maintain stirring. We derive rigorous lower bounds on both MKR mixing measures and $\dot H^{-1}$ norms showing that these cannot decay faster than exponentially in time. The bounds are in agreement with the exponential rates  observed during direct numerical simulations of optimal stirring, as reported in  \cite{LinThiffeaultDoering11}.

\medskip

The remainder of this paper is organized as follows: In Section \ref{S2} we introduce the MKR mixing measure and we present our main results. Section \ref{S3} contains the proofs. In the appendix, we finally discuss an important tool in our analysis, namely weighted local maximal functions.

\section{Mixing norms and main results}\label{S2}

Our primary result, the impossibility of perfect mixing in finite time, comes as a lower bound on an MKR mixing measure. Before stating the theorem, we provide the basic notation. Since MKR distances are the subject of the theory of optimal transportation, which is itself a wide and rapidly growing area of research, we will be unobtrusive in compiling general facts on MKR distances and refer the interested reader to the excellent monographs \cite{Villani03,Villani09}. 

\begin{definition}[Monge--Kantorovich--Rubinstein distance]\label{D2}
Let $\rho_0,\, \rho_1: [0,1)^d\to [0,\infty)$ be two periodic probability densities, i.e.,
\[
\int \rho_0\, dx\;=\;1\;=\;\int \rho_1\, dy.
\]
We denote by $\Gamma(\rho_0,\rho_1)$ the set of all probability measures on the product space $[0,1)^d\times [0,1)^d$ with marginals $\rho_0$ and $\rho_1$, i.e.,
\[
\iint \zeta(x)\, d\pi(x,y)\;=\; \int\zeta \rho_0\, dx, \quad \iint \zeta(y)\, d\pi(x,y)\;=\; \int\zeta \rho_1\, dy,
\]
for all periodic continuous functions $\zeta$ on $[0,1)^d$. Let $\sigma:[0,\infty)\to \R\cup\{-\infty\}$ denote a continuous function. The {\em Monge--Kantorovich--Rubinstein (MKR) distance} with cost function $\sigma$ is defined as
\[
d_c(\rho_0,\rho_1):=\inf_{\pi\in \Gamma(\rho_0,\rho_1)}\iint \sigma(|x-y|)\, d\pi(x,y).
\]
\end{definition}

Roughly speaking, the MKR distance $d_\sigma(\rho_0,\rho_1)$ measures the minimal total cost that is required for transferring mass from a source (producer) $\rho_0$ to a sink (consumer) $\rho_1$. The minimal total cost depends on the cost function $\sigma(z)$ which encodes the price one has to pay when transferring a unit of mass over a distance
 $z$. The measures $\pi\in\Gamma(\rho_0,\rho_1)$ are the transport plans, i.e., $\pi(x,y)$ is the amount of mass that is transported from $x$ to $y$.

\medskip

With the help of the MKR distance, we introduce the following mixing measure:

\begin{definition}[MKR mixing measure]\label{D3}
Let $\rho:[0,1)^d\to \{\pm1\}$ be a periodic function with
\[
\int \rho\, dx=0.
\]
Let $\rho_+$ and $\rho_-$ be the probability densities defined by
\[
\rho_+ =2\max\{\rho,0\}\quad\mbox{and}\quad\rho_-=-2\min\{\rho,0\},
\]
 so that $\rho=\frac12(\rho_+-\rho_-)$. We then define the {\em MKR mixing measure} as
\[
D(\rho):= \exp\left( d_{\ln}(\rho_+,\rho_-)\right).
\]
\end{definition}

Notice that even though the above definition can be easily extended to the vector space of mean-zero functions on $[0,1)^d$, the quantity $D(\rho)$ is not a norm as it lacks the triangle inequality.

\medskip

We remark that the above MKR mixing measure has units of length and describes thus, at least heuristically, the typical size of the unmixed regions in the mixture. 
Moreover, by the virtue of being defined via an MKR distance, $D(\rho)$ measures the total cost required for transporting $\rho_+$ into $\rho_-$, which is decreasing when the typical size of the unmixed regions gets smaller, that is, when the fluid is better mixed.

\medskip

Similar logarithmic MKR distances were used in \cite{BOS,OSS} in the study of demixing problems, where the flow field $u$ is slaved to the indicator function $\rho$ via a Stokes equation, modelling spinodal decomposition in binary viscous fluids. In a certain sense, our first result comes as a corollary of a dissipation inequality derived in \cite[Proposition 2.2]{BOS}.

\begin{theorem}[\cite{BOS}]\label{T1}
Let $1<p\le\infty$. There exists a constant $c>0$ depending on $p$ and $d$ only such that for every $T>0$
\[
D(\rho(T))\;\ge\; D(\rho_0) \exp\left(-c\int_0^T \|\grad u\|_{L^p}\, dt\right).
\]
\end{theorem}

The statement in Theorem \ref{T1} shows that the MKR mixing measure cannot decay faster than exponentially in time, and thus, perfect mixing in finite time is impossible for {\em any} given velocity vector field $u$ for which
\[
\int_0^T \|\grad u\|_{L^p}\, dt<\infty\quad\mbox{for all }T>0.
\]

\begin{remark}[Possible generalizations]\label{R1}
We remark there is room for several generalizations of the above result:

\begin{enumerate}
\item An analogous bound can be derived in the case where $m:=\int \rho\, dx\in(-1,1)\setminus\{0\}$. In this situation, the definition of $D(\rho)$ has to be modified: $\rho_{\pm}$ has to be replaced by $(\rho-m)_{\pm}$, where $(.)_{\pm}$ is defined as in Definition \ref{D3}.

\item As the system \eqref{8}\&\eqref{9} is invariant under the rescaling $x=L^{-1}\hat x$, $t=L^{-1}\hat t$, $\rho=\hat \rho$, and $u=\hat u$, our results extend to any domain $[0,L)^d$ with $L>0$. In this case, one has to understand the involved spatial integrals as spatial {\em averages}.

\item The assumption that $\rho\in\{\pm1\}$ is rather for notational convenience. Qualitatively similar results are true for any {\em smooth density} $\rho\in\R$ with $\|\rho\|_{L^{p'}}\le 1$ where $1/p+1/p'=1$.
\end{enumerate}
In any of these case, the constants appearing in the statement of Theorem \ref{T1} remain generic, i.e., they depend on $d$ and $p$ only.
\end{remark}

Notice that the statement in Theorem \ref{T1} is closely related to a mixing conjecture of Bressan \cite{Bressan03}. A version of this conjecture has been already established in \cite{CrippaDeLellis} and our approach uses techniques developed therein. The connection between Bressan's conjecture and $\dot H^{-1}$ decay rates is discussed in \cite[Sec.\ 4.A]{Lunasin12}.

\medskip

Although the MKR mixing measure is not known to be equivalent to a (negative) Sobolev norm, the following estimates, derived in Lemmas \ref{L1} and \ref{L2} below, indicate that MKR distances yield indeed reasonable measures of the degree of mixedness:
\begin{equation}\label{9aa}
\frac{c}{[\rho]_{BV}}\;\le\; D(\rho)\; \le\; \|\rho\|_{\dot H^{-1}}.
\end{equation}
Here, $c>0$ is a uniform constant that depends only on the space dimension and $[\rho]_{BV}$ is the $BV$ norm of the indicator function $\rho$. If the boundary of the level sets of $\rho$ are sufficiently regular, we simply have $[\rho]_{BV} = 2\Ha^{d-1}(\{\rho=1\})$.

\medskip

As a corollary of Theorem \ref{T1} and \eqref{9aa}, we obtain the following estimate on the decay rate of the $\dot H^{-1}$ norm.

\begin{theorem}\label{T2}
Let $1<p\le\infty$. There exists constants $c,C>0$ depending on $p$ and $d$ only such that for every $T>0$
\[
[ \rho_0]_{BV} \|\rho(T,\tacka)\|_{\dot H^{-1}}\;\ge\; C \exp\left(-c\int_0^T \|\grad u\|_{L^p}\, dt\right).
\]
\end{theorem}

This second statement excludes the possibility of finite time mixing in terms of the more classical $\dot H^{-1}$ mixing norm. The lower bound is in agreement with the robust numerical results obtained in \cite{LinThiffeaultDoering11}, where exponential decay rates were observed for the $\dot H^{-1}$ mixing norm when applying optimal stirring strategies. Unfortunately, our approach does not allow for a measurement of the degree of mixedness of the initial configuration in terms of this $\dot H^{-1}$ norm, but only in terms of the inverse $BV$ norm. The latter norm is weaker in the sense that $[\rho]_{BV}^{-1}\lesssim \|\rho\|_{\dot H^{-1}}$, see \eqref{9aa}. Notice that (inverse) $BV$ norms are suitable measures of demixing problems \cite{ContiNiethammerOtto06,BOS} rather than mixing problems, as small fluctuations along the boundary of $\{\rho=1\}$ increase the perimeter while the typical length scale of the unmixed regions (and thus the $\dot H^{-1}$ norm) remains unchanged. Hence, smallness of $[\rho]_{BV}^{-1}$ does not imply that the typical length scale is small. For sufficiently regular initial configuration (e.g.\ stripe or checkerboard pattern), however, both quantities are morally equivalent. The interplay of $BV$ and $\dot H^{-1}$ norms has been the subject of intensive research in recent years, see e.g.\ \cite{AlbertiChoksiOtto09}. 

\begin{remark}[Possible generalizations]
Again, we discuss possible generalizations of the statement of Theorem \ref{T2}.
\begin{enumerate}
\item With the same modifications as in Remark \ref{R1}, the analogous result applies when considering arbitrary systems $[0,L)^d$ with $L>0$.
\item In the case of off-critical mixtures, e.g., $1>m\gg0$, an analogous result is not obvious. Our proof of estimating the MKR mixing norm by the $\dot H^{-1}$ norm is quite crude as it relies essentially on Jensen's inequality. In the off-critical case, Jensen's inequality produces new $m$-dependent constants in both the rate $c$ and the factor $C$, since $(\rho-m)_{\pm}$ are not probability distributions. We expect, however, that a more careful analysis would still provide uniform decay rates, in analogy to the formally equivalent statement for the MKR mixing measures $D(\rho)$.  The constant $C$, however, will certainly show an $m$-dependence as a consequence of the sharp estimate \cite[Lemma 1.2]{ContiNiethammerOtto06}.
\item The result in Theorem \ref{T2} can be extended to smooth functions $\rho\in\R$ in several ways. One possibility is by replacing the $BV$ norm by the Ginzburg--Landau energy
\[
E(\rho)=\int \frac12|\grad \rho|^2 +\frac12(1-\rho^2)^2\, dx,
\]
and restricting to the low-energy regime $E(\rho)\le1$ (which includes perfect mixing $\rho\equiv0$), see also \cite[Prop.\ 4.4]{OSS}. In a certain rigorous sense ($\Gamma$-convergence), the Ginzburg--Landau energy is a smooth-interface regularization of the $BV$ norm \cite{ModicaMortola}. 
\end{enumerate}
\end{remark}

In the remainder of this paper, we prove Theorems \ref{T1} and \ref{T2}.

\section{Proofs}\label{S3}

\subsection{Proof of Theorem \ref{T1}}\label{S3.1}

Though the proof of Theorem \ref{T1} is a small variant of the proof of Proposition 2.2 in \cite{BOS}, we will provide most of its details for the sake of completeness and for the reader's convenience.

\begin{proof}[Proof of Theorem \ref{T1}] Since $u$ is divergence-free and $\rho_{\pm}=\pm(\rho\pm1)$, both $\rho_+$ and $\rho_-$ satisfy a continuity equation
\[
\partial_t \rho_{\pm} + \div(u\rho_{\pm})\;=\;0.
\]
Arguing as in \cite[Lemma 4.1]{BOS}, we can thus show that $t\mapsto d_{\ln_{\eps}}( \rho_+(t),\rho_-(t))$ is absolutely continuous and 
\[
\frac{d}{d t} d_{\ln_{\eps}}( \rho_+,\rho_-)\;\gtrsim\;- \iint \frac{| u(x)- u(y)|}{|x-y|}\, d\pi^{\eps}_*( t)(x,y),
\]
where $\ln_{\eps}(z)=\ln(z+\eps)-\ln\eps$ for some small $\eps>0$ and $\pi_*^{\eps}( t)$ denotes an optimal transport plan for the MKR distance $d_{\ln_{\eps}}( \rho_+( t), \rho_-( t))$. Here we have introduced a regularization of the cost function in order to ensure finiteness for all $z\ge0$. The idea how to estimate the above Lipschitz norm of $u$ by the $L^p$ norm of its gradient is inspired by the work of Crippa \& DeLellis \cite{CrippaDeLellis} and uses the notion of maximal functions. If $M:= M^0_1$ denotes the maximal function operator, see Definition \ref{D1}, it holds
\begin{equation}\label{9ab}
\frac{|f(x) - f(y)|}{|x-y|}\;\lesssim\; (M|\grad f|)(x) + (M|\grad f|)(y),
\end{equation}
and, for $p\in(1,\infty]$ 
\begin{equation}\label{9ac}
\| M|f|\|_{L^p}\lesssim\; \|f\|_{L^p}.
\end{equation}
The first estimate is proved in a slightly more general version in Lemma \ref{L1} of the appendix. The second estimate is the fundamental inequality for maximal functions, cf.\ \cite[p.\ 5, Theorem 1]{Stein70}. Applying \eqref{9ab}, we infer
\[
\frac{d}{d t} d_{\ln_{\eps}}( \rho_+,\rho_-)\;\gtrsim\;- \iint (M|\grad  u|)(x) + (M|\grad  u|)(y)\, d\pi_*^{\eps}( t)(x,y),
\]
which can be rewritten as
\[
\frac{d}{d t} d_{\ln_{\eps}}( \rho_+,\rho_-)\;\gtrsim\; -\int (M|\grad  u|) \rho_+\, dx  + \int (M|\grad  u|)\rho_-\, dy
\]
by the definition of the transport plans. Now, using $\|\rho_{\pm}\|_{L^{\infty}}=2$, Jensen's inequality and \eqref{9ac}, we infer that
\[
\frac{d}{d t} d_{\ln_{\eps}}( \rho_+,\rho_-)\;\gtrsim\;- \| \grad u \|_{L^p}.
\]
Integrating in time now yields
\[
d_{\ln_{\eps}}( \rho_+(T),\rho_-(T)) \;\ge\; d_{\ln_{\eps}}( (\rho_0)_+,(\rho_0)_-)  - c\int_0^T \|\grad u\|_{L^p}\, dt,
\]
for some constant $c>0$ depending only on $d$ and $p$. Since transport plans are probability measures, the above estimate implies
\[
\iint \ln(|x-y|+\eps)\, d\pi(T)(x,y) \;\ge\; \iint \ln(|x-y|+\eps)\, d\pi_*^{\eps}(0)(x,y) - c\int_0^T \|\grad u\|_{L^p}\, dt,
\]
for all $\pi(T)\in\Gamma(\rho_+(T),\rho_-(T))$. By the monotonicity of the logarithm, it is
$\iint \ln(|x-y|+\eps)\, d\pi_*^{\eps}(0)(x,y)\ge d_{\ln}((\rho_0)_+, (\rho_0)_-)$, and thus, passing to the limit $\eps\downarrow 0$ we have
\[
\iint \ln|x-y|\, d\pi(T)(x,y)\;\ge\; d_{\ln}((\rho_0)_+, (\rho_0)_-) - c\int_0^T \|\grad u\|_{L^p}\, dt,
\]
for all $\pi(T)\in\Gamma(\rho_+(T),\rho_-(T))$. Minimizing over $\pi(T)$ and applying the exponential function yields the statement of the theorem.
\end{proof}

\subsection{Proof of Theorem \ref{T2}}\label{S3.2}

The statement in Theorem \ref{T2} is an immediate consequence of Theorem \ref{T1} and the following two lemmas:

\begin{lemma}\label{L1}
For every periodic function $\rho: [0,1)^d\to \R$ with $\int \rho\, dx=0$, it holds that
\[
D(\rho)\;\le\; \|\rho\|_{\dot H^{-1}}.
\]
\end{lemma}

\begin{lemma}\label{L2} There exists a constant $c>0$ depending only on the space dimension, such that for every $\rho: [0,1)^d\to \{\pm 1\}$ with $\int \rho\, dx=0$ it holds
\[
c\;\le\;  D(\rho)[\rho]_{BV}.
\]
\end{lemma}
The first lemma shows that the MKR mixing measure is dominated by the $\dot H^{-1}$ norm. The second lemma provides a lower bound on the MKR mixing measure in terms of the inverse $BV$ norm and is derived as an interpolation inequality. In a certain sense, the estimate can be seen as a non-standard isoperimetric inequality.

\medskip

We start with the relatively simple 
\begin{proof}[Proof of Lemma \ref{L1}]
Let $\pi\in\Gamma(\rho_+,\rho_-)$ be arbitrary, so that
\[
D(\rho)\;\le\; \exp\left(\iint\ln|x-y|\, d\pi\right).
\]
We use Jensen's inequality together with the fact that $\pi$ is a probability measure to infer
\[
D(\rho)\;\le\; \iint |x-y|\, d\pi,
\]
and thus, minimizing over all $\pi\in \Gamma(\rho_+,\rho_-)$ yields
\[
D(\rho)\;\le\; d_{id}(\rho_+,\rho_-).
\]
By the Kantorovich--Rubinstein theorem (cf.\ \cite[Theorem 1.14]{Villani03}), $d_{id}(\rho_+,\rho_-)$ is a dual norm, namely
\[
d_{id}(\rho_+,\rho_-) = \sup\left\{\int \rho\zeta\, dx:\: \|\grad\zeta\|_{L^{\infty}}\le1\right\}.
\]
Because $\|\grad\zeta\|_{L^2}\le\|\grad\zeta\|_{L^{\infty}}$, we thus deduce that
\[
D(\rho)\;\le\;\sup\left\{\int \rho\zeta\, dx:\: \|\grad\zeta\|_{L^{2}}\le1\right\}\\
\;=\; \|\rho\|_{\dot H^{-1}}.
\]
\end{proof}

The proof of the interpolation inequality of Lemma \ref{L2} is more involved. It comes as a nontrivial extension of the estimate \cite[Prop. 2.3]{BOS}. In the new proof below, we use {\em weighted local maximal functions}, that will be introduced and discussed in the appendix.

\begin{proof}[Proof of Lemma \ref{L2}]
Let $0<r\le 1$ and $R>0$ be two arbitrary constants that will be fixed at the end of the proof. We denote by $\rho_R$ the convolution of $\rho$ with a radially symmetric standard mollifier that is supported on a ball of radius $R$. We then write
\begin{equation}\label{4}
1\;=\;\int \rho^2\, dx \; =\; \int \rho(\rho-\rho_R)\,dx + \int \rho\rho_R\, dx.
\end{equation}
The estimate of the first integral on the right is standard and can easily be derived via approximation by smooth functions, namely
\begin{equation}\label{4a}
\int \rho(\rho-\rho_R)\,dx\;\le\; \|\rho\|_{L^{\infty}}R\int |\grad\rho|\;=\; 2R[\rho]_{BV}.
\end{equation}
For the second term, we choose an arbitrary transport plan $\pi\in\Gamma(\rho_+,\rho_-)$ and write
\begin{eqnarray}
\lefteqn{\int \rho\rho_R\, dx}\nonumber\\
&=& \frac12\int \rho_R(\rho_+ - \rho_-)\, dx\nonumber\\
&=& \frac12\iint (\rho_R(x)-\rho_R(y))\, d\pi\nonumber\\
&=& \frac12\iint_{|x-y|\ge r} (\rho_R(x)-\rho_R(y))\, d\pi + \frac12\iint_{|x-y|< r} (\rho_R(x)-\rho_R(y))\, d\pi.\label{5}
\end{eqnarray}
Since $\pi$ is a probability measure, the second term in \eqref{5} can be controlled as follows:
\begin{equation}\label{4b}
\iint_{|x-y|< r} (\rho_R(x)-\rho_R(y))\, d\pi\;\le\;\|\grad \rho_R\|_{L^{\infty}} r\iint d\pi\;\lesssim\;\|\rho\|_{L^{\infty}} \frac{r}{R}\;=\;\frac{r}{R}.
\end{equation}
It remains to estimate the first term in \eqref{5}. We first notice that $1\le \ln |x-y| +\ln\frac{e}{r}$ for all $|x-y|\ge r$, and thus
\begin{eqnarray}
\lefteqn{\iint_{|x-y|\ge r} (\rho_R(x)-\rho_R(y))\, d\pi}\nonumber\\
 &\le& 2\|\rho_R\|_{L^{\infty}} \iint_{|x-y|\ge r}d\pi\nonumber\\
&\le& 4\iint \left(\ln |x-y| + \ln\frac{e}{r}\right)\, d\pi + 4\iint_{|x-y|<r}\ln\frac{r}{e|x-y|} \, d\pi\nonumber\\
&=& 4\iint \ln |x-y| \, d\pi+ 4\ln\frac{e}{r}+ 4\iint_{|x-y|<r}\ln\frac{r}{e|x-y|}\, d\pi.\label{6}
\end{eqnarray}
We now have
\begin{eqnarray*}
\lefteqn{\iint_{|x-y|<r}\ln\frac{r}{e|x-y|}\, d\pi}\\
&=&\iint_{|x-y|<r}\ln\frac{e^2 r}{|x-y|} \, d\pi - 3\iint_{|x-y|<r}d\pi\\
&\le& \frac12 \iint_{|x-y|<r}\sqrt{|x-y|}\left(\ln\frac{e^2 r}{|x-y|} \right)\frac{|\rho(x)-\rho(y)|}{\sqrt{|x-y|}}\, d\pi,
\end{eqnarray*}
because $\rho(x)=1$ and $\rho(y) =-1$ for all $(x,y)$ in the support of $\pi$. Observe that $s\mapsto\sqrt s \ln\frac{e^2r}{s}$ is an increasing function on $(0,r]$, and thus
\[
\iint_{|x-y|<r}\ln\frac{r}{e|x-y|}\, d\pi\;\le\; \sqrt r \iint_{|x-y|<r}\frac{|\rho(x)-\rho(y)|}{\sqrt{|x-y|}}\, d\pi.
\]
The integral on the right can be estimated with the help of weighted local maximal functions, that will be introduced and discussed in the appendix. For this purpose, let $\{\rho_{\eps}\}_{\eps\downarrow0}$ be a smooth approximation of $\rho$.
Applying Lemma \ref{L3} below with $f=\rho_{\eps}$ and $\theta=\frac12$ yields
\[
\iint_{|x-y|<r}\frac{|\rho_{\eps}(x)-\rho_{\eps}(y)|}{\sqrt{|x-y|}}\, d\pi\;\lesssim\; \iint_{|x-y|<r} \left((M_r^{1/2}|\grad \rho_{\eps}|)(x) + (M_r^{1/2}|\grad \rho_{\eps}|)(y)\right)\, d\pi,
\]
and thus, by the definition of $\pi$ and since $\|\rho_{\pm}\|_{L^{\infty}}=2$
\[
\iint_{|x-y|<r}\frac{|\rho_{\eps}(x)-\rho_{\eps}(y)|}{\sqrt{|x-y|}}\, d\pi\;\lesssim\; \int |M_r^{1/2}|\grad \rho_{\eps}||\, dx.
\]
We now invoke Lemma \ref{L4} with $f=|\grad \rho_{\eps}|$ and $\tau=\frac12$, to the effect of
\[
\iint_{|x-y|<r}\frac{|\rho_{\eps}(x)-\rho_{\eps}(y)|}{\sqrt{|x-y|}}\, d\pi\;\lesssim\;   \sqrt{r} \int |\grad \rho_{\eps}|\,dx .
\]
Passing to the limit $\eps\downarrow0$ via Fatou's lemma, the above calculation shows
\begin{equation}\label{7}
\iint_{|x-y|<r}\ln\frac{r}{e|x-y|}\, d\pi\;\lesssim\; r[\rho]_{BV}.
\end{equation}
Finally, combining \eqref{4}--\eqref{7} yields
\[
1\;\le\; 2R[\rho]_{BV} +C \frac{r}{R} + 2\iint \ln|x-y|\, d\pi +2\ln\frac{e}{r} + Cr[\rho]_{BV},
\]
for some $C>0$. We first choose $R=Cr$ to the effect of
\[
0\;\le\; \iint \ln|x-y|\, d\pi + \ln\frac{e}{r} + \frac32Cr[\rho]_{BV}.
\]
Now, optimizing in $r$ yields $r\sim \frac1{[\rho]_{BV}}$ (observe that $[\rho]_{BV}\gtrsim1$ by the virtue of the isoperimetric inequality) and obtain, upon applying the exponential function to both sides of the inequality
\[
c\;\le\; \exp\left(\iint\ln|x-y|\, d\pi\right) [\rho]_{BV},
\]
for some new constant $c>0$. The statement of the lemma follows when optimizing over $\pi\in\Gamma(\rho_+,\rho_-)$.
\end{proof}

\subsection*{Appendix: Weighted local minimal functions}\label{SA}
The proof of Lemma \ref{L2} above uses the notion of weighted local maximal functions that will be introduced and discussed in the following.

\begin{definition}[Weighted local minimal functions] \label{D1}
Let $f:[0,1)^d\to [0,\infty)$ denote a smooth periodic function, and let $0<r\le 1$ and $0\le \tau\le 1$ be two constants. We define the {\em weighted local maximal function} of $f$ as
\begin{equation}\label{9a}
(M_r^{\tau}f)(x):= \sup_{0<q<r} \frac1{|B_q(0)|}\int_{B_q(0)} |z|^{\tau} f(x+z)\, dx.
\end{equation}
\end{definition}

Observe that we recover classical local maximal functions in the case $\tau=0$, cf.\ \cite[Ch.\ 1]{Stein70}. The following estimate is well-known to hold in the case $\theta=1$ (or $\tau=0$).

\begin{lemma}\label{L3}Let $f:[0,1)^d\to \R$ denote a smooth periodic function, and let $0<r\le 1$ and $0< \theta\le 1$ be two constants. Then there exists a constant $C>0$ depending only on $d$ and $\theta$ such that
\begin{equation}\label{1}
\frac{|f(x)-f(y)|}{|x-y|^{\theta}}\;\le\; C\left( (M_r^{1-\theta}|\grad f|)(x)  + (M^{1-\theta}_r |\grad f|)(y)\right),
\end{equation}
for all $x,y\in[0,1)^d$ such that $|x-y|\le r$.
\end{lemma}

\begin{proof}Let $Q:=|x-y|$. We start with a classical argument that replaces a pointwise by an averaged quantity: We first show that
\[
\frac{|f(x)-f(y)|}{|x-y|^{\theta}}\;\lesssim\; \frac1{Q^d}\int_{B_Q(x)} \frac{|f(x)-f(z)|}{|x-z|^{\theta}}\, dz + \frac1{Q^d}\int_{B_Q(y)} \frac{|f(y)-f(z)|}{|y-z|^{\theta}}\, dz.
\]
Indeed, via the triangle inequality and averaging over $B_Q(x)\cap B_Q(y)$, we see that
\begin{eqnarray*}
\lefteqn{|f(x) - f(y)|}\\
&\lesssim& \frac1{Q^d}\int_{B_Q(x)\cap B_Q(y)}|f(x)-f(z)|\, dz + \frac1{Q^d}\int_{B_Q(x)\cap B_Q(y)}|f(y)-f(z)|\, dz\\
&\le&\frac1{Q^d}\int_{B_Q(x)}|f(x)-f(z)|\, dz + \frac1{Q^d}\int_{ B_Q(y)}|f(y)-f(z)|\, dz,
\end{eqnarray*}
and the statement follows from $|x-z|\le Q=|x-y|$ for all $z\in B_Q(x)$. For the statement in \eqref{9a} it thus suffices to show that
\begin{equation}\label{2}
\frac1{Q^d}\int_{B_Q(x)} \frac{|f(x)-f(z)|}{|x-z|^{\theta}}\, dz \;\lesssim\; (M_r^{1-\theta}|\grad f|)(x),
\end{equation}
for all $x\in[0,1)^d$. Without loss of generality, we may assume that $x=0$. Then, using the Fundamental Theorem of Calculus, we estimate
\[
\frac1{Q^d}\int_{B_Q(0)} \frac{|f(0)-f(z)|}{|z|^{\theta}}\, dz \;\le\; \frac1{Q^d}\int_{B_Q(0)} |z|^{1-\theta} \int_0^1 |\grad f(sz)|\, dsdz.
\]
Furthermore, via Fubini and a change of variables, this becomes
\begin{eqnarray*}
\frac1{Q^d}\int_{B_Q(0)} \frac{|f(0)-f(z)|}{|z|^{\theta}}\, dz &\le&\int_0^1 \frac1{s^{1-\theta}}\frac1{(sQ)^d}\int_{B_{sQ}(0)} |\hat z|^{1-\theta} |\grad f(\hat z)|\, d\hat zds\\
&\le&\left(\int_0^1 \frac1{s^{1-\theta}}\, ds\right) (M_r^{1-\theta}|\grad f|)(0),
\end{eqnarray*}
where in the last inequality we have used $sQ\le sr\le r$ and the definition of weighted local maximal functions. In order to deduce \eqref{2}, it remains to observe that the prefactor is bounded for $\theta>0$.
\end{proof}

We finally derive the fundamental inequality for weighted local maximal functions. Observe that this estimate fails in the case $\tau=0$, that is, for classical maximal functions.

\begin{lemma}\label{L4}
Let $f:[0,1)^d\to [0,\infty)$ denote a smooth periodic function, and let $0<r\le 1$ and $0< \tau\le 1$ be two constants. Then there exists a constant $C>0$ depending only on $d$ and $\tau$ such that
\begin{equation}\label{3}
\int |M_r^{\tau} f |\, dx\;\le\; C r^{\tau}\int |f|\, dx.
\end{equation}
\end{lemma}

\begin{proof}
For every $0< q\le r$, we define
\[
F(x,q):= \frac1{q^d} \int_{B_q(0)} |z|^{\tau} f(x+z)\, dz.
\]
Because $\tau>0$, we have that $\lim_{|q|\downarrow0}F(\tacka,q)=0$, and thus the Fundamental Theorem of Calculus yields
\begin{eqnarray*}
\lefteqn{F(x,q)}\\
 &=& \int_0^q \frac{\partial}{\partial Q}F(x,Q)\, dQ\\
&=& \int_0^q\left(-\frac{d}{Q^{d+1}}\int_{B_Q(0)}|z|^{\tau} f(x+z)\, dz + \frac1{Q^d} \int_{\partial B_Q(0)} |z|^{\tau} f(x+z)\, d\Ha^{d-1}(z)\right)\, dQ\\
&\le&\int_0^r \frac1{Q^d}\int_{\partial B_Q(0)} |z|^{\tau} f(x+z)\, d\Ha^{d-1}(z)\, dQ,
\end{eqnarray*}
where in the last inequality we have used that $q\le r$ and $f\ge0$. Via a chance of variables $z=Q\hat z$, we obtain the estimate
\[
F(x,q)\;\lesssim\;  \int_0^r Q^{\tau-1} \int_{\partial B_1(0)}  f(x + Q\hat z)\, d\Ha^{d-1}(\hat z)dQ.
\]
We take the supremum over all $0<q\le r$ on the left hand side and integrate over $[0,1)^d$. Using Fubini and the periodicity of $f$, we obtain
\begin{eqnarray*}
\int |M_r^{\tau} f|\, dx&\lesssim& \int \int_0^r Q^{\tau-1} \int_{\partial B_1(0)}  f(x + Q\hat z)\, d\Ha^{d-1}(\hat z)dQdx\\
&=&   \int_0^rQ^{\tau-1}\, dQ  \int_{\partial B_1(0)}  d\Ha^{d-1}(\hat z) \int |f|\, dx\\
&\lesssim&r^{\tau} \int |f|\, dx.
\end{eqnarray*}
Notice that we have used $\tau>0$ in the last estimate. This proves Lemma \ref{L4}
\end{proof}

\section*{Acknowledgment}The author thanks Charlie Doering and Evelyn Lunasin for bringing this topic to his attention, for enlightening discussions, and for the kind hospitality at the University of Michigan, Ann Arbor. He also acknowledges motivating discussions with Felix Otto.

\bibliography{coarsening}
\bibliographystyle{acm}
\end{document}